\documentclass[letterpaper,10pt,conference]{ieeeconf}
\IEEEoverridecommandlockouts 
\pdfoutput=1

\usepackage{xparse}
\usepackage{color}
\usepackage{booktabs}
\usepackage{multirow}
\usepackage{multicol}
\usepackage{graphicx} 

\graphicspath{{figs/}}

\usepackage{amsmath}
\usepackage{amssymb}
\usepackage{amsfonts}

\usepackage{amsthm}

\theoremstyle{definition}
\newtheorem{definition}{Definition}
\newtheorem{lemma}{Lemma}
\newtheorem{theorem}{Theorem}

\usepackage{cite}
\usepackage{balance}
\usepackage{cleveref}
\crefformat{equation}{(#2#1#3)}
\Crefname{figure}{Fig.}{Figs.}
\IEEEtriggeratref{9}

\newcommand{\bbm}{\begin{bmatrix}}
\newcommand{\ebm}{\end{bmatrix}}
\newcommand{\vect}[1]{\boldsymbol{\mathbf{#1}}}
\newcommand{\matr}[1]{\boldsymbol{\mathbf{#1}}}

\newcommand{\Real}{\ensuremath{\mathbb{R}}}

\newcommand{\SimpleSystem}{\ensuremath{ \Sigma }}
\newcommand{\ForwardReachableSet}[3]{ \mathcal{F}(#1, #2, #3) }
\newcommand{\BackwardReachableSet}[3]{ \mathcal{B}(#1, #2, #3) }

\newcommand{\taut}{\ensuremath{\tau}}   
\newcommand{\tauf}{\ensuremath{\tau'}}  

\title{\Large\bf A Question of Time: Revisiting the Use of Recursive\\ Filtering for Temporal Calibration of Multisensor Systems}

\author{Jonathan Kelly$^{\dagger}$, Christopher Grebe, and Matthew Giamou$^{\ddagger}$ 
\thanks{All authors are with the Space \& Terrestrial Autonomous Robotic Systems (STARS) Laboratory at the University of Toronto Institute for Aerospace Studies (UTIAS), Toronto, Canada. {\tt\footnotesize <firstname>.<lastname>@robotics.utias.utoronto.ca}}
\thanks{$^\dagger$Jonathan Kelly is a Vector Institute Faculty Affiliate. This research was supported in part by the Canada Research Chairs program.}
\thanks{$^{\ddagger}$Matthew Giamou is a Vector Institute Postgraduate Affiliate and an RBC Fellow.}}

\begin{document}
\maketitle

\begin{abstract}
We examine the problem of time delay estimation, or temporal calibration, in the context of multisensor data fusion.
Differences in processing intervals and other factors typically lead to a relative delay between measurement updates from disparate sensors.
Correct (optimal) data fusion demands that the relative delay must either be known in advance or identified online.
There have been several recent proposals in the literature to determine the delay using recursive, causal filters such as the extended Kalman filter (EKF). 
We carefully review this formulation and show that there are fundamental issues with the structure of the EKF (and related algorithms) when the delay is included in the filter state vector as a parameter to be estimated.
These structural issues, in turn, leave recursive filters prone to bias and inconsistency.
Our theoretical analysis is supported by simulation studies that demonstrate the implications in terms of filter performance; although tuning of the filter noise variances may reduce the chance of inconsistency or divergence, the underlying structural concerns remain.
We offer brief suggestions for ways to maintain the computational efficiency of recursive filtering for temporal calibration while avoiding the drawbacks of the standard filtering algorithms.
\end{abstract}

\section{Introduction}
\label{sec:intro}

Time delays are inherent in all sensor data processing due to, for example, signal conditioning, integration, and transmission latency.
Optimal data fusion from multiple sources requires that the relative delays between sensor data streams must either be determined offline in advance or identified online, when possible.

In this paper, we consider the task of estimating a single delay in a multisensor system.
A variety of systematic time delay identification, or temporal calibration, methods have been published in the literature---there are bespoke, offline estimators for specific sensor pairs \cite{2011_Mair_Spatio-Temporal,2014_Kelly_General,2014_Kelly_Determining}, more general batch programs \cite{2013_Furgale_Unified,2016_Rehder_General}, and online, sliding-window implementations \cite{2018_Qin_Online}, for example.
Our focus herein is on recent proposals to determine the delay parameter through recursive, causal filtering, using variants of the extended Kalman filter (EKF).
In particular, Nilsson, Skog, and H\"{a}ndel \cite{2010_Nilsson_Joint,2011_Skog_Time} have applied the EKF to state and delay estimation for GPS-aided inertial navigation, while Li and Mourikis have devised a similar approach for vision-aided inertial navigation \cite{2013_Li_3-D,2014_Li_Online}.
The delay parameter is added to the filter state vector and the measurement update step is modified to account for the delay uncertainty.
Filtering has several potential advantages when implemented successfully, including simplicity and computational efficiency. 

Although recursive filtering is appealing for time delay estimation, we show that there are several problems with this formulation.
Our exposition is based upon the analysis of a simplified aided navigation system, where measurements from one sensor are treated as control inputs to drive the filter process model while updates are provided by a second, complementary sensor.
An overview of the delay model is provided in \Cref{fig:overview}.

\begin{figure}[t]
\vspace{1mm}
\centering
\includegraphics[width=\columnwidth]{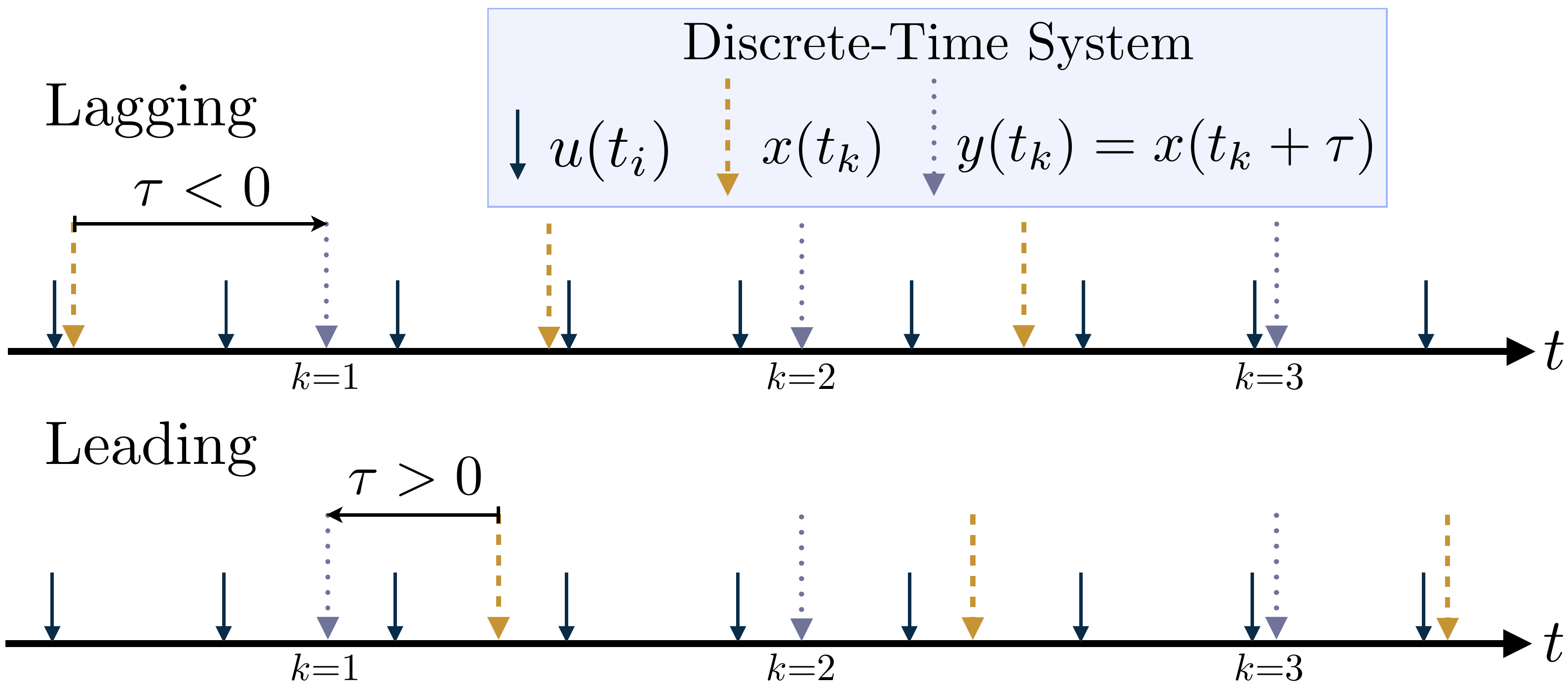}
\caption{A simple time-delay system.
Measurements $u(t_{i})$ from a reference sensor are treated as control inputs to drive the system process model forward in time.
A second sensor provides updates $y(t_{k})$ at regular intervals.
There is an unknown, relative delay $\tau$ between the sensor data streams, such that the update $y(t_{k}) = x(t_{k} + \tau)$ either lags or leads the state, $x(t_{k})$.
\emph{Top:} If $\tau$ is negative, $y(t_{k})$ is delayed relative to $x(t_{k})$, that is, lagging the state.
\emph{Bottom:} If $\tau$ is positive, $y(t_{k})$ is advanced relative to $x(t_{k})$, that is, leading the state.
Either situation may occur in practice.}
\label{fig:overview}
\vspace{-3mm}
\end{figure}

We begin by proving that the delay parameter requires a finite time for identification and by establishing that the point estimate computed by the EKF does not properly account for this requirement.
We then examine how the addition of process noise is coupled to the shifting delay estimate.
Modifications to the delay value change the interval between measurement updates, altering the way in which process noise is incorporated within the EKF.
We demonstrate experimentally that, because of these structural issues, recursive filters such as the EKF are prone to bias and inconsistency (and, potentially, to divergence).
Critically, these fundamental problems exists due to the structure of the estimator---although they can be `covered up' to some extent by tuning of the filter noise variances, the underlying structural concerns remain.
Our contributions are as follows:
\begin{enumerate}
	\item we make explicit the requirements for delay identifiability and explain how these requirements are not satisfied by the EKF, leading to the erroneous computation of the filter innovation sequence;
	\item we establish that the calculated uncertainty of the filter measurement update does not represent the true (underlying) uncertainty in general, making the filter prone to bias and inconsistency;
	\item we show that discontinuous `jumps' (in time) occur at each filter update step, due to changes in the delay estimate, and that these jumps adversely affect the way in which process noise is incorporated into the algorithm; and
	\item we demonstrate through extensive simulations that the issues above manifest in practice and can easily compromise the performance of an EKF-type estimator.
\end{enumerate}

We confine our analysis to the simplified system model for clarity and brevity, but we emphasize that the same shortcomings exist for more complex models, including those used for aided inertial navigation, and for other recursive filters (e.g., the unscented Kalman filter and various particle filters).
While we do not prescribe a full remedy in this short paper, we close by offering several directions to explore that may lead to a structurally sound, recursive algorithm for simultaneous state and delay estimation.  

\section{Background}
\label{sec:background}

A vast body of research on estimation and control under time delays exists.
We give only a brief synopsis here of material that lies within our scope, emphasizing  proximal work. Filtering-based methods for delay estimation are described in \Cref{sec:filter}.
An overview of some open problems in time-delay systems is provided in \cite{2003_Richard_Time-Delay}.

Recursive filtering with delayed and out-of-sequence measurements is a well-studied problem that can be solved optimally (for linear systems) if the delays are known \cite{2006_Simon_Optimal}.
In the case where measurements arrive out of sequence, the additional complication of correlated noise terms must be dealt with, as shown by \cite{1994_Thomopoulos_Decentralized}.
Handling retrodicted measurements in these situations is computationally expensive; an alternative, proposed in \cite{1998_Larsen_Incorporation}, is to extrapolate measurements forward in time, sacrificing optimality for computational tractability.

The problem of recursive state (and parameter) estimation becomes substantially more difficult when the measurement delays are unknown and constant or subject to random variation.
Assuming that the delays are random but quantized and bounded by some known value, Julier and Uhlmann \cite{2005_Julier_Fusion} show that the method of \emph{covariance union} can be applied to produce consistent, filtered state estimates.
Covariance union is a necessarily conservative algorithm, however, that achieves consistency at the expense of estimator efficiency when the delays can be better characterized.
If the delays are random but the delay distribution is available a priori, Choi et al.\ describe a state augmentation technique in \cite{2009_Choi_State} to produce consistent estimates. 
Past states are added to the state vector (increasing the computational complexity) and measurements are then `allocated' across the past and current state estimates according to the delay distribution.

When possible, proper delay identification, carried out either offline or online, can improve the performance of recursive estimators. 
Identification of systematic delays is typically a challenging task, particularly in the  nonlinear setting. 
Xia et al.\ \cite{2002_Xia_Analysis} and Zhang, Xia, and Moog \cite{2006_Zhang_Parameter} have characterized parameter and state identifiability for nonlinear time-delay systems using the theory of non-commutative rings. 
Anguelova and Wennberg extend this framework in \cite{2007_Anguelova_State,2008_Anguelova_State} by providing linear-algebraic criteria to determine whether multiple delays are identifiable.
Importantly, the results in \cite{2008_Anguelova_State} reveal that the identifiability of  delays is not directly related to the identifiability or observability of other system parameters or states, respectively.

There is relevant prior research on deterministic and stochastic observer design for a variety of time-delay systems \cite{2001_Sename_New}.
For certain classes of systems, observers with desirable convergence properties, such as global exponential convergence, can be constructed.
Many of these observers are so-called `chain observers' (see, e.g., \cite{2002_Germani_New} and \cite{2005_Kazantzis_Nonlinear}) that attempt to reconstruct the state at a series of points within the delay window.
The treatment of stochastic time-delay systems is substantially more involved and fewer convergence results are available.
We return to the topic of stochastic observer design in \Cref{sec:conclusion}.

\section{Identifiability of Time Delays}
\label{sec:ident}


In the remainder of the paper, we examine a very simple control-affine system,
\begin{equation}
\label{eqn:single_state}
\SimpleSystem : \quad
\begin{aligned}
\dot{x}(t) & = u(t),
\quad x(0) = x^{0},\\	
y(t) & = x(t + \tau),
\end{aligned}
\end{equation}
%
where $x(t) \in \Real$ is the system state at time $t$, $u(t) \in \Omega = [-1, 1]$ is the applied control, $y(t) \in \Real$ is the system output, and $\tau \in \Real$ is the delay parameter.
Our notation roughly follows that in \cite{1977_Hermann_Nonlinear}. We label specific points on the manifold $\Real$ using superscripts (this avoids a clash with subscripted time indices in \Cref{sec:filter}).
The state of the system at $t = 0$ is $x(0) = x^{0}$.
If $\tau$ is negative, the output is delayed relative to the state (or \emph{lags} the state); if $\tau$ is positive, the output is advanced relative to the state (or \emph{leads} the state).
We use the term \emph{delay} whether $\tau$ is positive or negative, when there is no risk of ambiguity, and indicate the absolute value of the delay by $\tau_{+}$.
Throughout this paper we consider $\tau$ to be fixed but unknown. 
Note that, when no delay is present (i.e., for $\tau = 0$ exactly), the system \cref{eqn:single_state} is trivially observable.

Our goal in this section is to show clearly that the delay parameter of \Cref{eqn:single_state} requires a finite time for identification.
This is an intuitive and straightforward result that is implicit in much of the prior and proximal work. 
We make the result explicit in order to connect it to difficulties involved in recursive filtering later in \Cref{sec:filter}.

\subsection{An Identifiability Criterion}
\label{subsec:ident_criterion}

We take the standard approach to determining identifiability (also sometimes called \emph{structural identifiability} in this context \cite{1970_Bellman_Structural}) based on output distinguishability, as defined in \cite{1976_Grewal_Identifiability}.
The definitions below have been adjusted to fit our problem instance where necessary.

We denote an admissible control input on the closed interval $[t^{0}, t^{1}]$ by the pair $\left(u(t), [t^{0}, t^{1}]\right)$ and the set of all admissible controls on the interval $(-\infty, \infty)$ by $\mathcal{U}$.
With a slight abuse of set notation, we denote the concatenation of two sequential controls that share an endpoint in time with the union operator, for example,
$\left(u(t), [t^{0}, t^{1}]\right) \cup \left(0, (t^{1}, t^{2}]\right)$.

We assume that the system exists for all time but that its output can only be observed for a finite duration, starting arbitrarily at $t = 0$, when information about the control input $u(t)$ is also available.
In certain cases, the state of the system, $x(0)$, at $t = 0$ may also be known.
To determine the delay parameter, we may perform any number of experiments, where an `experiment' involves observing the output $y(t)$ given knowledge of $u(t)$ and that $x(0) = x^{0}$.
We use $y(t, \tau)$ to denote the output of \Cref{eqn:single_state} at time $t$ when the system has delay $\tau$.
The requisite definitions follow.

\begin{definition}[Indistinguishability]
The pair of delay parameter values $(\taut$, $\tauf)$ are said to be \emph{indistinguishable} if 
\begin{equation}
\label{eqn:indistinguishable}
y(t, \taut) = y(t, \tauf)
\end{equation}
for all $x(0) = x^{0}$ and for any admissible control $u(t)$ over the interval $0 \leq t \leq T$.
Otherwise, the two delay values are said to be \emph{distinguishable} over the interval.
\end{definition}

\begin{definition}[Identifiability]
The delay parameter value $\taut$ is \emph{identifiable} if the pair $(\taut, \tauf)$ is distinguishable for all $\tauf$ such that $\tau \neq \tauf$.
\end{definition}

It is important to recognize, in the definitions above, that although the outputs are observed over the \emph{same} interval $[0, T]$, those outputs correspond to \emph{different} times for the underlying system (due to differing delays).

\subsection{Limits on Identifiability}
\label{subsec:ident_limits}

We first define the time-limited forward and backward reachable sets \cite{2006_LaValle_Planning,2017_Bansal_Hamilton-Jacobi} of \Cref{eqn:single_state} from a given state.\footnote{In the control literature, these are usually called the reachable tubes; we use the name \emph{sets} in our context, following \cite{2006_LaValle_Planning}.}

\begin{lemma}
\label{lem:fb_reachable}
For the system $\Sigma$, the time-limited forward reachable set from state $x^0$ in time $T \geq 0$ is the closed interval $\ForwardReachableSet{x^{0}}{\mathcal{U}}{T} := \left[x^{0} - T, x^{0} + T\right]$.
The time-limited backward reachable set from state $x^0$ in time $T$ is $\BackwardReachableSet{x^{0}}{\mathcal{U}}{T} := \left[x^{0} - T, x^{0} + T\right]$.
\end{lemma}

\begin{proof}
The results follow directly from the definition of the set of available controls, $\Omega$, and the time, $T$.
\end{proof}

We now state an important theorem that will be central to our discussion of recursive filtering in \Cref{sec:filter}. 
The main argument of the theorem is illustrated in \Cref{fig:proof_plot}.

\begin{theorem}
\label{thm:tau_unobservable}
The delay parameter, $\tau$, of the system $\Sigma$ cannot be identified in a time less than or equal to $\tau_{+}$.
\end{theorem}

\begin{proof}
The proof is by construction.
We show that there always exist two instances of \cref{eqn:single_state}, with different delays, that share the same input-output map for any admissible control $\big(u(t), [0, \tau_{+}]\big)$.
Denote these two instances by $\Sigma$ and $\Sigma'$.\footnote{We apply the `primed' notation throughout this section to distinguish the two systems that have differing delay values.}
\smallskip

Consider the delays $\taut, \tauf \in \Real_{<0}$ such that $\tauf < \taut$. 
Let the state of the system $\Sigma$ with delay $\taut$ be $x(0) = x^{0}$. Likewise, let the state of the system $\Sigma'$ with delay $\tauf$ be $x'(0) = x^{0}$. 
An arbitrary control $\big(u(t), [0, \tau_{+}]\big)$ applied to both systems results in trajectories that are identical on $[0, \tau_{+}]$.
However, the applied control does not affect the outputs, $y(t)$ and $y'(t)$, of $\Sigma$ and $\Sigma'$, respectively, over the interval $[0, \tau_{+}]$, because the measured outputs are of \emph{past} states.

Let $B = \BackwardReachableSet{x^{0}}{\mathcal{U}}{\tau_{+}} := \left[x^{0} + \tau, x^{0} - \tau\right]$ and $B' = \BackwardReachableSet{x^{0}}{\mathcal{U}}{\tauf_{+}} := \left[x^{0} + \tauf, x^{0} - \tauf\right]$ be the time-limited backward reachable sets from $x^{0}$ for $\Sigma$ and $\Sigma'$, respectively, by \Cref{lem:fb_reachable}.
Since the time-limited backward reachable set from $x^{0}$ is not a singleton for either $\Sigma$ or $\Sigma'$, knowledge of $x^{0}$ fails to distinguish $\Sigma$ from $\Sigma'$.\footnote{If the system admits only one possible trajectory, then knowledge of the state at any time uniquely determines the state at all past and future times.}

Let the outputs be $y(0) = x(\taut) = y'(0) = x'(\tauf) = x^{-1}$, where $x^{-1} \in B \cap B'$.
This is always possible because $B \subset B'$.
Apply a control $u(t)$ to $\Sigma$ that is a solution to $\dot{x}(t) = u(t)$ with $x(\taut) = x^{-1}$ and $x(0) = x^{0}$, and apply the control $\big(u'(t), [\tauf, 0]\big) =   \big(u(t + \taut - \tauf), [\tauf, \tauf - \taut]\big) \cup \big(0, (\tauf - \taut, 0]\big)$ to $\Sigma'$ with $x'(\tauf) = x^{-1}$ (see the inputs in \Cref{fig:proof_plot}).
We have
\begin{multline}
x^{-1} + \int^{0}_{\taut} u(t)\,dt = \\
x^{-1} + \int^{\tauf - \taut}_{\tauf} u(t + \taut - \tauf)\,dt + \int^{0}_{\tauf - \taut} (0)\,dt = x^{0}.
\end{multline}
The trajectories of $\Sigma$ and $\Sigma'$ are identical over the intervals $[\taut, 0]$ and $[\tauf, \tauf - \taut]$. Hence, the outputs are also identical over $[0, \taut_{+}]$ and $\taut$ and $\tauf$ cannot be distinguished.
\medskip

\begin{figure}[t]
\centering
\includegraphics[width=\columnwidth,trim=0 17spt 0 0,clip]{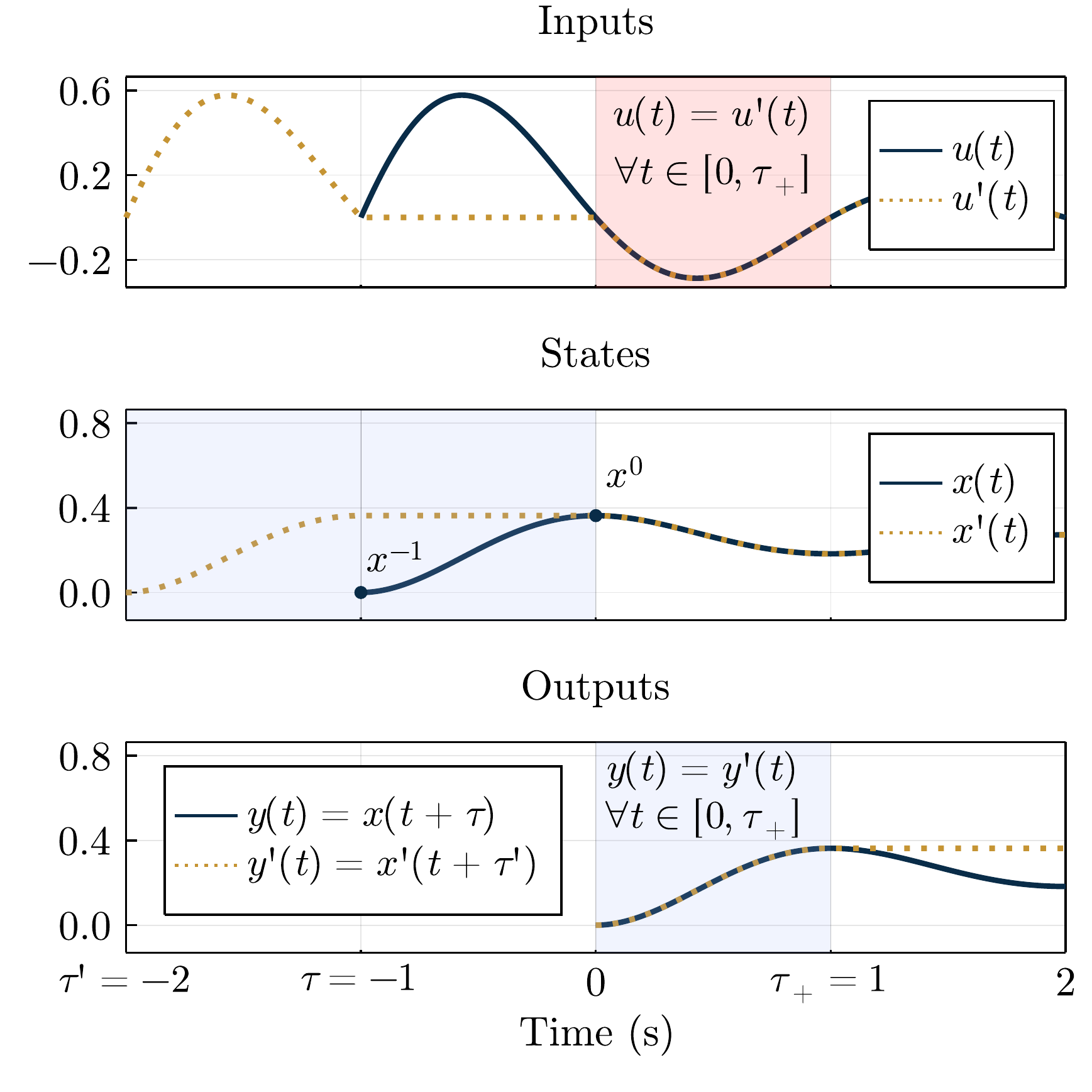}
\vspace*{-8mm}
\caption{An illustration of the main argument in the proof of \Cref{thm:tau_unobservable} for $\tau = -1$ and $\tau' = -2$. The control $u(t)$ on interval $[0, \tau_{+}]$ (shaded red) does not affect the output $y(t)$ (shaded blue) on $[0, \tau_{+}]$. Therefore, for the inputs $u(t)$ and $u'(t)$ illustrated here, $\tau$ and $\tau'$ are indistinguishable on $[0, \tau_{+}]$ and the time delay is unidentifiable.}\label{fig:proof_plot}
\vspace*{-4mm}
\end{figure}

Similarly, consider the delays $\tau, \tau' \in \Real_{>0}$ such that $\tau' > \tau$.
As before, let $x(0) = x'(0) = x^{0}$ and let an arbitrary control 
$\big(u(t), [0, \tau_{+}]\big)$ be applied to both systems, resulting in trajectories that are identical on $[0, \tau_{+}]$.
The applied control does not affect the outputs, $y(t)$ and $y'(t)$, of $\Sigma$ and $\Sigma'$, respectively, over the interval $[0, \tau_{+}]$, because the measured outputs are of \emph{future} states.

Let $F = \ForwardReachableSet{x^{0}}{\mathcal{U}}{\tau} := \left[x^{0} - \tau, x^{0} + \tau\right]$ and $F' = \ForwardReachableSet{x_{0}}{\mathcal{U}}{\tau'} := \left[x^{0} - \tau', x^{0} + \tau'\right]$ be the time-limited forward reachable sets from $x^{0}$ for $\Sigma$ and $\Sigma'$, respectively, by \Cref{lem:fb_reachable}.

Let the outputs be $y(0) = x(\taut) = y'(0) = x'(\tauf) = x^{1}$, where $x^{1} \in F \cap F'$.
This is always possible because $F \subset F'$.
Apply a control $u(t)$ to $\Sigma$ that is a solution to $\dot{x}(t) = u(t)$ with $x(\taut) = x^{1}$ and $x(2\taut) = x^{2}$, and apply the control 
$\big(u'(t), [\taut, \tauf + \taut]\big) = \big(0, [\taut, \tauf)\big) \cup \big(u(t + \taut - \tauf), [\tauf, \tauf + \taut]\big)$ to $\Sigma'$ with $x'(\tauf) = x^{1}$.
We have
\begin{multline}
x^{1} + \int^{2\taut}_{\taut} u(t)\,dt = \\
x^{1} + \int^{\tauf}_{\taut} (0)\,dt + \int^{\tauf + \taut}_{\tauf} u(t + \taut - \tauf)\,dt = x^{2}.
\end{multline}
The trajectories of $\Sigma$ and $\Sigma'$ are identical over the intervals $[\taut, 2\taut]$ and $[\tauf, \tauf + \taut]$. Hence, the outputs are again identical over the interval $[0, \tau_{+}]$ and $\taut$ and $\tauf$ cannot be distinguished.
\end{proof}

The proof involves a (possibly) discontinuous control function but this is for convenience only---the result immediately generalizes to continuous controls as well.

\section{Recursive Filtering}
\label{sec:filter}

In this section, we explore how the problem of delay identifiability manifests in recursive filtering. 
Our theoretical analysis is supported by the experimental results in \Cref{sec:experiments}.
We treat \Cref{eqn:single_state} as a simple model of an aided navigation system and refer to the sensor whose measurements drive the process model as the \emph{reference sensor}.

\subsection{Extended Kalman Filter Formulation}
\label{subsec:ekf}

We apply the modified continuous-discrete (hybrid) extended Kalman filtering algorithm described in \cite{2014_Li_Online} to the system \Cref{eqn:single_state}.
The filter state vector is
\begin{equation}
\label{eqn:state_vector}
\vect{x}(t) = \bbm x(t) & \tau \ebm^{T},
\end{equation}
where we call $x(t)$ the \emph{position at time $t$}.
The state evolves in continuous time according to
\begin{equation}
\label{eqn:dynamics}
\dot{x}(t) = u(t) + w(t),\quad 
\dot{\tau}(t) = 0.
\end{equation}
We take the usual approach and define the noise term $w(t) \sim \mathcal{GP}\big(0, Q_{x}\delta(t - t')\big)$, that is, as a zero-mean, white noise Gaussian process with the variance function $Q_{x}\delta(t - t')$, where $Q_{x}$ is the power spectral density and $\delta(\cdot)$ is the Dirac delta function.
The state mean and state covariance are propagated forward in time using
\begin{equation}
\dot{\hat{x}}(t) = u(t),\quad
\dot{\hat{\tau}}(t) = 0,\quad
\dot{\matr{P}}(t) = \matr{Q},
\end{equation}
where $\matr{Q}$ is a $2 \times 2$ matrix with $Q_{x}$ in the upper left and the remaining entires set to zero. Following the standard notation, we use the $\hat{\cdot}$ (hat) to indicate an estimated quantity.

Measurement updates occur at discrete time steps, identified by the index variable $k = 1, \dots, n$, for $n$ measurements in total.
The measurement $y_{k}$ is
\begin{equation}
\label{eqn:ekf_meas}
y_{k} = h\left(\vect{x}_{k}\right) + v_{k} = x(t_{k} + \tau) + v_{k},
\end{equation}
where $v_{k} \sim \mathcal{N}\big(0, R\big)$ is an additive, zero-mean Gaussian noise term with variance $R$.
All of the noise processes are assumed to be statistically independent at all times.

In order to compute the innovation (i.e., the measurement residual) and the Kalman gain, we require the Jacobian of \Cref{eqn:ekf_meas} with respect to the state vector.
This step is complicated by the fact that $\tau$ is a random variable. 
Following \cite{2014_Li_Online}, we integrate the state and state covariance forward in time over a varying interval until the next measurement arrives.
In a sampled-data system, the most recent control input (i.e., from the reference sensor) drives this integration.
At time step $k$, the best estimate of the delay is $\hat{\tau}_{k - 1}$.
We therefore linearize \Cref{eqn:ekf_meas} about this operating point, which gives
\begin{equation}
\matr{H}_{k} =
\left.\frac{\partial h}{\partial \mathbf{x}}\right|_{\hat{\mathbf{x}}^{-}_{k}} = 
\bbm
1 & \dot{\hat{x}}(t_{k} + \hat{\tau}_{k - 1})
\ebm =
\bbm
1 & u(t_{k} + \hat{\tau}_{k - 1})
\ebm.
\end{equation}
We use the superscript minus and plus signs to denote an estimated quantity immediately before and immediately after the measurement update, respectively.
Unlike the standard EKF formulation, the measurement Jacobian in this case is a function of the most recently-sampled control value, because the control provides the only available information about the rate of change of the state.

As final steps in the recursive EKF measurement update, we compute the innovation variance, the Kalman gain, the updated filter state estimate, and the updated filter state covariance matrix,
\begin{align}
\label{eqn:innovation_var}
S_{k} & = \matr{H}_{k}\matr{P}_{k}^{-}\matr{H}_{k}^{T} + R, \\[1mm]
\label{eqn:ekf_gain}
\matr{K}_{k} & = 
\matr{P}_{k}^{-}\matr{H}_{k}^{T}
S^{-1}_{k}, \\
\label{eqn:ekf_state_update}
\hat{\vect{x}}_{k}^{+} & =
\hat{\vect{x}}_{k}^{-} + 
\matr{K}_{k}\Big(y_{k}-h_{k}\left(\hat{\vect{x}}_{k}^{-}\right)\Big), \\
\label{eqn:ekf_covar_update}
\matr{P}_{k}^{+} & =
\big(\matr{I} - \matr{K}_{k}\matr{H}_{k}\big)\matr{P}_{k}^{-},
\end{align}
where $y_{k}$ is the most recent measurement to arrive.

\subsection{Time as Measured by the Filter}
\label{subsec:update_time}

According to \Cref{eqn:state_vector}, the filtered state estimate is maintained at time $t$ (defined by the reference sensor clock).
Critically, however, there is an ambiguity in \Cref{eqn:ekf_state_update} and \Cref{eqn:ekf_covar_update} regarding the time at which the EKF update is actually applied.
At each iteration, the state is integrated forward in time to $t_{k} + \hat{\tau}_{k - 1}$, prior to the measurement update at time step $k$, and $\hat{\tau}$ is then \emph{modified} by the update.
As a result, the estimated arrival time of incoming (future) measurements is `shifted' by some amount. In turn, the state estimate maintained by the filter is not for the time $t_{k}$, rather the estimate is for the time $t_{k} + \hat{\tau}_{k}$ (immediately after a measurement update).
This distinction (and shifting behaviour) is relevant to the discussion of process noise in \Cref{subsec:state_estimate}.

\subsection{Estimation of the Delay Parameter}
\label{subsec:delay_estimate}

In \Cref{sec:ident} we established that the delay parameter requires a finite time for identification.
Here, we consider the question of whether the innovation variance, defined by \Cref{eqn:innovation_var}, represents the true uncertainty of the innovation.
The scalar term
\begin{equation}
\label{eqn:state_innovation_variance}
S_{s, k} = \matr{H}_{k}\matr{P}_{k}^{-}\matr{H}_{k}^{T}
\end{equation}
in \Cref{eqn:innovation_var} is the projection of the \emph{state} uncertainty into the measurement space.
Recall that the measurement Jacobian $\matr{H}_{k}$ depends on $u(t_{k} + \hat{\tau}_{k - 1})$, the most recent control.

There is a structural concern that arises immediately when using \Cref{eqn:state_innovation_variance} to compute the variance: the applied control at time $t_{k} + \hat{\tau}_{k - 1}$ may be very different from the applied control at an earlier or later time, and hence the state of the system may be very different also.
To illustrate this, consider the case where, on filter startup, $\hat{\tau}_{0} = 0$ and $u(0) \approx 0$.
The true delay is unknown and hence when the first measurement arrives, the position of the system, $x(t + \tau)$, might be anywhere in the forward or backward reachable set from $x(0) = x^{0}$, if $x(0)$ is known.
We can write \Cref{eqn:state_innovation_variance} as
\begin{equation}
S_{s, 1} \approx
\bbm
1 & u(0)
\ebm
\bbm
\sigma_{x, 1}^{2} & 0 \\
0 & \sigma_{\tau, 0}^{2}
\ebm
\bbm
1 & u(0)
\ebm^{T}
\approx
\sigma_{x, 1}^{2},
\end{equation}
where $\sigma_{x, 1}$ is the standard deviation of the position uncertainty just prior to the measurement update and $\sigma_{\tau, 0}$ is the standard deviation of the time delay uncertainty.
Because the control input has a value near zero, the uncertainty in the delay is down-weighted very significantly when computing $S_{s, k}$, leading to an innovation variance that is too small and to a spurious Kalman gain calculation.\footnote{The opposite case, where the innovation variance is too large, can also occur, but this is potentially less of an issue.}
Further, this situation is not confined to the first measurement update---it appears at every update step (see \Cref{subsec:recursion}).

The innovation variance should reflect the true uncertainty at each iteration, based on the reachable sets for the system, but this is not the case in general for \Cref{eqn:state_innovation_variance} (nor, therefore, for \Cref{eqn:innovation_var} either).
Since the Gaussian distribution has short tails, there is a risk that the innovation variance will frequently be too small (i.e., that the Gaussian will be too `narrow') leading to an inconsistent filter estimate.
We show in \Cref{sec:experiments} that this situation arises in practice.
It is important to emphasize that this structural problem is not the result of linearization error in the model. 
Rather, the problem is due to an unavoidable lack of knowledge about the state and the control inputs within the delay interval.
A possible `solution' is to inflate the measurement uncertainty, but this introduces a range of other concerns.

\subsection{Estimation of the System State and Covariance}
\label{subsec:state_estimate}

The results from Section III also have implications for the position (state) estimate and the position uncertainty.
If the delay estimate is wrong, this will tend to have detrimental effect on the state estimate (because the state and the delay become correlated as the filter operates).
The usual way to handle modelling errors is to tune the filter process noise (co)variance, inflating the filter state covariance matrix as a measure of ignorance of the underlying, `true' model. 

While the tuning of process noise is often effective, there is a second structural concern that makes tuning difficult in the case of delay estimation: discontinuous `jumps' in the delay estimate lead to jumps in the time maintained by the filter and alter the way in which process noise is incorporated.

Because the delay estimate is never exactly correct, each measurement update is always applied at an incorrect point in time.
And because every measurement update alters the value of $\hat{\tau}$, there is a change in the integration period over which process noise is added.
It follows that the integration period is never exactly correct and the amount of process noise added is also never `correct,' relative to the process noise that would be added in a system where the delay was known exactly.

Depending on the sequences of shifts in $\hat{\tau}$, from $k - 1$ to $k$ and so on, the sum of the integration periods may be too long or too short.
The process noise variance can only be adjusted up or down as part of the tuning process and so both cases cannot be handled simultaneously.
Here, the primary problem is that the specific sequence of `jumps' is dependent on the dynamics and the random noise---again, the worst outcome is filter inconsistency.
Once the filter estimate is inconsistent, the only recourse is to restart the estimator.

Another interesting consequence of the structural problem described above is the possibility of `time inversion' within the filter.
Consider the estimated delay $\hat{\tau}_{k - 1}$ after the measurement update at time step $k - 1$ and the estimated delay $\hat{\tau}_{k}$ after the measurement update at time step $k$.
There is nothing that prevents the change $\Delta\hat{\tau} = \hat{\tau}_{k} - \hat{\tau}_{k - 1}$ from being a (relatively large) \emph{negative} quantity (this can occur due to, e.g., the effects of random measurement noise).
Following the discussion in \Cref{subsec:update_time}, if the interval between measurements is short, na\"{i}ve application of the filter time update can require the process model to operate \emph{backwards} in time.
The standard implementation of the EKF (and other Bayes filters) adds process noise when the filter runs forwards in time; it is not immediately clear how to handle the paradoxical (reverse) situation in a rigorous way.

\subsection{Effect of Recursion}
\label{subsec:recursion}

The EKF maintains an estimate at a single point in time.
Since the filter model is a Markov chain, there is no memory beyond the state-delay cross-covariance.
The estimate of the delay parameter, $\hat{\tau}$, will never be exactly correct and so the structural issues (viz.\ identifiability/observability and `jumps' of the state in time) described in \Cref{subsec:delay_estimate,subsec:state_estimate} appear \emph{on every iteration} of the filtering algorithm.

\section{Simulation Studies}
\label{sec:experiments}

In this section, we apply the algorithm from \cite{2014_Li_Online} to the system \Cref{eqn:single_state} and evaluate the performance of the estimator through Monte Carlo simulations.
We consider the setting of \Cref{sec:filter}, where the filter process model operates in continuous time while measurements arrive at discrete time steps (i.e., the hybrid EKF).

\subsection{Experiment Design}
\label{subsec:design}

We chose two test trajectories with different characteristics for our evaluation.
The equation of each trajectory has a smooth analytic form that is the sum of simple sinusoids, allowing exact derivatives to be computed. 
Trajectory 1, shown at the top of \Cref{fig:traj1_results}, is a `conservative' trajectory---the rate of change of the velocity (i.e., the control input $u(t)$ to the system) is limited.
Trajectory 2, shown at the top of \Cref{fig:traj2_results}, is a more aggressive trajectory with a larger maximum acceleration.
Note that nonzero acceleration is a requirement in this case: if the velocity remains constant, the system state is unobservable and the delay parameter is unidentifiable regardless of the observation duration.\footnote{This requirement is easy to show. Under constant velocity $u(t) = v$, an arbitrary time delay $\tau$ can always be compensated for by a shift in the initial position state by $\Delta x = -v\tau$, yielding the same output for $y(t)$.}

\begin{figure}[t]
\centering
\includegraphics[width=\columnwidth,trim=0 15pt 0 0,clip]{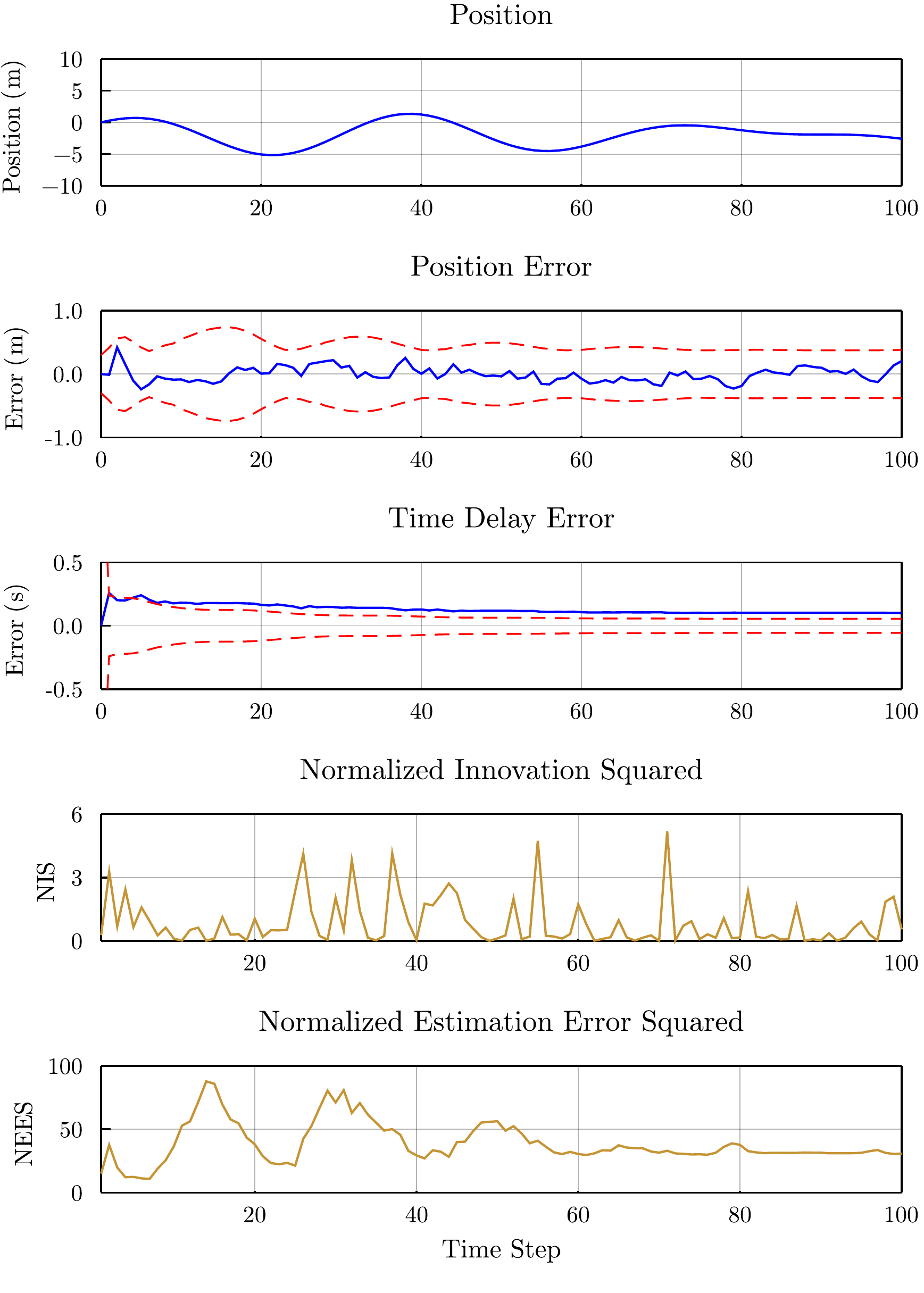}
\caption{Results from one Monte Carlo simulation. \emph{Top:} Trajectory 1. \emph{Middle:} Position and time delay errors (blue solid lines) and associated 3$\sigma$ bounds (red dashed lines) with respect to time step. \emph{Bottom:} Filter NIS and NEES with respect to time step. In this case, the filter is inconsistent and displays apparent divergence. The true delay value was $\tau = 5.9$ ms.} 
\label{fig:traj1_results}
\vspace*{-3mm}
\end{figure}

\begin{figure}[t]
\centering
\includegraphics[width=\columnwidth,trim=0 15pt 0 0,clip]{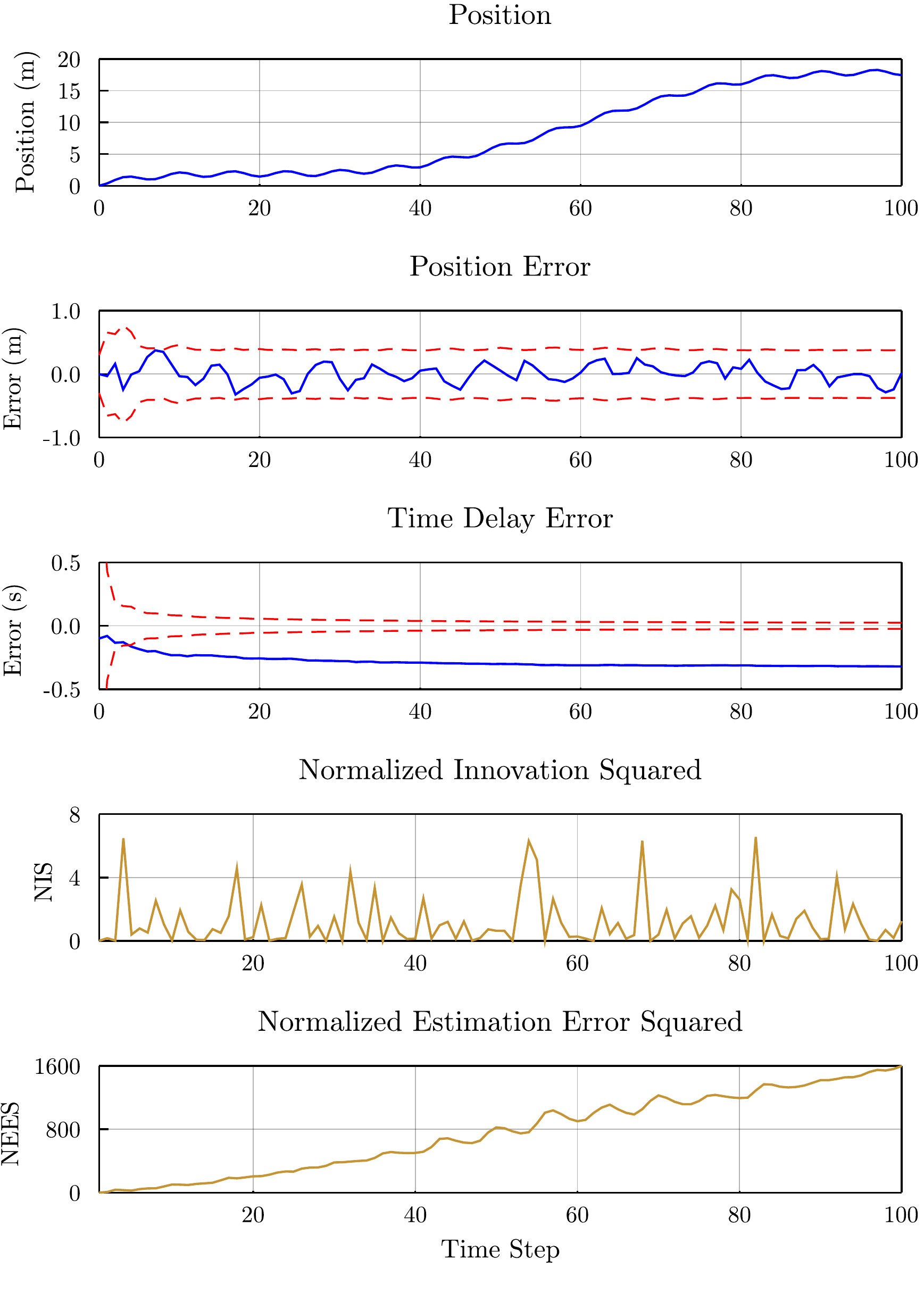}
\caption{Results from one Monte Carlo simulation. \emph{Top:} Trajectory 2. \emph{Middle:} Position and time delay errors (blue solid lines) and associated 3$\sigma$ bounds (red dashed lines) with respect to time step. \emph{Bottom:} Filter NIS and NEES with respect to time step. The delay estimate rapidly becomes inconsistent and ultimately diverges (beyond the 10-second window shown). The true delay value was $\tau = -84.2$ ms.}
\label{fig:traj2_results}
\vspace*{-3mm}
\end{figure}

We ran 1,000 Monte Carlo simulations for each trajectory, randomly varying the value of the (true) time delay on a per-trial basis.
The delay value was sampled from a zero-mean Gaussian distribution with a standard deviation of 50 ms.\footnote{Although we specify simulation quantities in metres and seconds for convenience, this choice is arbitrary.}
Position updates occurred at a rate of 10 Hz (i.e., every 100 ms); each position measurement was corrupted by zero-mean Gaussian noise with a standard deviation of 0.25 m.
The continuous-time power spectral density matrix \cite{2008_Farrell_Aided} for the process noise was
\begin{equation}
\matr{Q} = \bbm 1.0~\text{m$^{2}$/s$^{2}$} & 0 \\ 0 & 0 \ebm,
\end{equation}
which was tuned for the experiments to ensure that the position estimate did not diverge.

At the start of each trial, the state vector and state covariance matrix were initialized to
\begin{align}
\hat{\vect{x}}_{0} & = \bbm 0~\text{m} & 0~\text{s} \ebm^{T}, \\[1mm]
\hat{\matr{P}}_{0} & = \bbm 0.01~\text{m$^{2}$} & 0 \\ 0 & 0.25~\text{s$^{2}$} \ebm.
\end{align}
An initial choice of $\hat{\tau} = 0~\text{s}$ is reasonable if no prior information about the delay is available.
Importantly, the uncertainty of the initial delay estimate was intentionally set to $\sigma_{\tau} = 0.5~\text{s}$ (i.e., $\sigma^{2}_{\tau} = 0.25~\text{s}^2$).
This value is an order of magnitude larger than the standard deviation of the Gaussian distribution from which the true delay values were sampled.
We found that, without inflating the initial delay uncertainty, the filter would almost always become rapidly inconsistent or diverge;
our inflation of the delay uncertainty is in line with existing work (e.g., \cite{2020_Geneva_OpenVINS}).

Based on the initial results, we ran an additional set of 1,000 simulation trials (per trajectory) where we fixed the true delay value at $\tau = -50~\text{ms}$ (i.e., a 50 ms lag in the measured position).
We then evaluated the consistency of the estimator as described in the next section.

\subsection{Performance Evaluation}
\label{subsec:performance}

We assessed the performance of the filter using a standard set of criteria intended to measure estimator accuracy and consistency.
Accuracy can be determined by calculating the average root mean square position and delay  errors across a series of time steps, over multiple simulation trials.
To analyze estimator consistency, we compute two statistics, the normalized innovation squared (NIS) and the normalized estimation error squared (NEES). The NIS and NEES are defined by the quadratic forms
\begin{equation*}
\label{eqn:nees}
\epsilon_{y, k} \triangleq 
\mathrm{e}_{y, k}^{T}\,S_{k}^{-1} \mathrm{e}_{y, k}~~\text{and}~~
\epsilon_{\vect{x}, k} \triangleq 
\mathbf{e}_{\vect{x}, k}^{T} \left(\matr{P}^{+}_{k}\right)^{-1} \mathbf{e}_{\vect{x}, k},
\end{equation*}
respectively, where $\mathrm{e}_{y, k} = y_{k}-h_{k}\left(\hat{\vect{x}}_{k}^{-}\right)$ is the innovation and $\mathbf{e}_{\vect{x}, k} = \vect{x}_{k} - \hat{\vect{x}}_{k}$ is the state estimate error, both calculated at time step $k$ \cite{2001_Bar-Shalom_Estimation}.
If the estimator is consistent, $\epsilon_{y, k}$ and $\epsilon_{\vect{x}, k}$ should be $\chi^{2}$ distributions with one and two degrees of freedom, respectively.
We follow the usual practice of computing the average NEES (ANEES) over multiple trials \cite{2001_Bar-Shalom_Estimation,2018_Chen_Weak} to obtain a better measure of performance.

Results from the Monte Carlo simulations with varying delays are summarized in \Cref{tab:rmsposition,tab:rmsdelay}.
Several trends are apparent: while the RMS position error remains bounded (in part due to tuning of the process noise), the RMS error for the time delay is substantial on average even by time step 20.
In fact, the delay error after 20 to 40 time steps is \emph{larger} on average than the initial error.

We also determined the number of times that the error in the filter estimate of the time delay was within the computed $3\sigma$ bounds after 100 time steps (i.e., 10 seconds of simulation time); this occurred for only 189 out of 1,000 trials for Trajectory 1 and for just 3 out of 1,000 trials for Trajectory 2.
Although Trajectory 2 is more `informative' (i.e., has greater excitation) than Trajectory 1, Table 2 indicates that the RMS delay error for Trajectory 2 is substantially larger than for Trajectory 1.

In \Cref{fig:traj1_results} and \Cref{fig:traj2_results}, results from one simulation trial of Trajectory 1 and one simulation trial of Trajectory 2 are shown, respectively.
\Cref{fig:traj1_results} displays the phenomenon of \emph{apparent divergence} of the delay estimate, that is, convergence to an incorrect, biased value \cite{1974_Gelb_Applied}.
By the 10-second mark, the inconsistency is severe enough that the filter is insensitive to additional data---even if the simulation is run for 100 seconds, the delay estimate does not improve (or converge).
The outcome is the same for Trajectory 2, although the delay estimate is significantly worse.
In both cases, a large number of the NIS values are much greater than one (confirming the analysis in \Cref{subsec:delay_estimate}).
Importantly, the performance suggested by these examples is better than for many of the random trials.

Results from the Monte Carlo simulations with the fixed time delay of $\tau = -50~\text{ms}$ are summarized in \Cref{tab:nees}.
The ANEES statistics show that the filter is frequently highly inconsistent for both trajectories.  
These results support the arguments from theory in \Cref{sec:filter} that the structure of the filter is unsound.

\subsection{Discussion}
\label{subsec:discussion}

Our experimental results verify that fundamental problems exist with the structure of the EKF as a recursive estimator for time delays.
In addition to being prone to bias and inconsistency, the filter is highly sensitive to  initial conditions and to random noise.
There is a complex interplay between the system dynamics, the delay magnitude, and the noise terms that makes tuning difficult.
Although we analyzed the simple system \Cref{eqn:single_state} only, the same structural problems exist for delay estimation in standard GNSS-aided and vision-aided inertial navigation systems and in other multisensor systems.
We posit that results presented in the literature to date, which show convergence of delay estimates, may not hold in the general case of varying dynamics and without careful tuning (that may not be possible).

\begin{table}[t]
\centering
\caption{RMS position error versus time step. Each entry is the RMS error over 1,000  simulation trials with varying delays.}
\label{tab:rmsposition}
\setlength{\tabcolsep}{10pt}
\renewcommand{\arraystretch}{1}
\begin{tabular*}{1.0\columnwidth}{@{\extracolsep{\fill}} p{17mm} c c c c c}
\toprule
& \multicolumn{5}{c}{\textbf{RMS Position Error [m]}} \\[1mm]
~~~~Time Step & 20 & 40 & 60 & 80 & 100 \\ 
\midrule
~~~~Trajectory 1 & 
 0.122 &
 0.126 &
 0.114 &
 0.102 &
 0.102 \\[1mm]
~~~~Trajectory 2 &
 0.107 &
 0.117 &
 0.109 &
 0.106 &
 0.112 \\
\bottomrule
\end{tabular*}
\end{table}

\begin{table}[b]
\centering
\caption{RMS delay error versus time step. Each entry is the RMS error over 1,000  simulation trials with varying delays.}
\label{tab:rmsdelay}
\setlength{\tabcolsep}{8pt}
\renewcommand{\arraystretch}{1}
\begin{tabular*}{1.0\columnwidth}{@{\extracolsep{\fill}} p{17mm} c c c c c}
\toprule
& \multicolumn{5}{c}{\textbf{RMS Delay Error [ms]}} \\[1mm]
~~~~Time Step & 20 & 40 & 60 & 80 & 100 \\ 
\midrule
~~~~Trajectory 1 & 
  81 &
 111 &
 127 &
 131 &
 132 \\[1mm]
~~~~Trajectory 2 & 
 128 &
 151 &
 165 &
 175 & 
 182 \\
\bottomrule
\end{tabular*}
\end{table}

\begin{table}[t]
\centering
\caption{Average NEES versus time step. Each entry is the ANEES over 1,000  simulation trials with the fixed delay $\tau = -50~\text{ms}$.}
\label{tab:nees}
\setlength{\tabcolsep}{8pt}
\renewcommand{\arraystretch}{1}
\begin{tabular*}{1.0\columnwidth}{@{\extracolsep{\fill}} p{17mm} c c c c c}
\toprule
& \multicolumn{5}{c}{\textbf{ANEES}} \\[1mm]
~~~~Time Step & 20 & 40 & 60 & 80 & 100 \\ 
\midrule
~~~~Trajectory 1 &
  10.6 &
  24.1 &
  45.5 &
  54.1 &
  54.6 \\[1mm]
~~~~Trajectory 2 & 
  54.9 &
 137.5 &
 254.5 &
 376.3 &
 520.4 \\
\bottomrule
\end{tabular*}
\end{table}

\section{Conclusion}
\label{sec:conclusion}

We have examined the application of recursive, causal filtering (in particular, the EKF) to state and delay estimation, as proposed recently in the literature.
Our analysis of the identifiability of the delay parameter for a simple time-delay system implies that the standard error calculation employed in the filter update step is specious.
In turn, the corrected state and delay estimates are prone to inconsistency.
Further, the shifts in the value of the delay parameter alter the way in which process noise is incorporated within the filtering framework.
We emphasize that these characteristics are \emph{structural} and not due to incorrect tuning of noise terms or to a lack of knowledge of the system dynamics.
Our experimental results confirm that the EKF formulation is very sensitive to initial conditions and generally produces inconsistent and biased delay estimates (in addition to poorer-quality state estimates).
Time-varying delays or measurement jitter may further degrade estimator performance.

While we do not suggest an immediate remedy for the problems above, there are several possibilities to explore.
One option is to employ a modified, sliding-window estimator (similar to \cite{2018_Qin_Online}), where the window size is sufficient to `cover' any feasible delay interval.
The price for this approach is the computational overhead of the windowed algorithm and the inherent lag in the state and parameter estimates.
Another solution could be to examine the design of a stochastic observer (e.g., a chain observer) for a specific system of interest;
it may be possible to derive valuable convergence results in certain cases. We are currently reviewing extensions to invariant filtering on the appropriate Lie group (e.g., $\mathrm{SE}(3)$ in the case of aided inertial navigation).

\section*{Acknowledgements}
\label{sec:ack}

We thank the anonymous reviewers for their valuable comments and suggestions that helped us to improve the clarity of this manuscript.

\bibliographystyle{IEEEcaps}
\bibliography{2021-kelly-temporal-mfi}

\end{document}